\documentclass[3p,times]{elsarticle}

\usepackage[utf8]{inputenc}
\usepackage{amssymb, amsthm, amsmath}
\usepackage[english]{babel}
\usepackage[T1]{fontenc}
 \usepackage{graphics}
 \usepackage{graphicx}
\usepackage{mathtools}
\usepackage[font=small,labelfont={normal, bf}]{caption}
\usepackage{subcaption}
\usepackage{color}

\usepackage[normalem]{ulem}

\usepackage[linesnumbered,ruled,vlined]{algorithm2e}

\newtheorem{theorem}{Theorem}
\newtheorem{proposition}{Proposition}
\theoremstyle{definition}
\newtheorem{definition}{Definition}

\usepackage{hyperref}
\hypersetup{
colorlinks=true
}
\usepackage{epstopdf}
\makeatother
\makeatletter
\newcommand{\eqdef}{\mathrel{\mathop:}=}
\makeatother

\biboptions{sort&compress}

\journal{Journal of Computational and Applied Mathematics}
\title{Cyclic fractional Gaussian noise: time and frequency domain properties}
\author[inst1]{Hubert Woszczek \corref{cor1}}
\author[inst1]{Agnieszka Wyłomańska}

\cortext[cor1]{Corresponding author. Email: hubert.woszczek@pwr.edu.pl}
 \affiliation[inst1]{organization={Faculty of Pure and Applied Mathematics, Hugo Steinhaus Center, Wroclaw University of Science and Technology},
            addressline={Hoene-Wronskiego 13c},
            city={Wroclaw},
            postcode={50-376},
            country={Poland}}
                      
\date{}
\begin{document}

\begin{frontmatter}



\title{}

\begin{abstract}
 This article introduces cyclic fractional Gaussian noise (cfGn), a stochastic model that integrates second-order cyclostationarity with long-range dependence property. While classical cyclostationary processes are widely discussed in the literature, they often lack the capacity to account for the persistent, slow-decaying correlations found in complex empirical data. To bridge this gap, we extend the amplitude-modulated stationary framework by utilizing increments of two-dimensional fractional Brownian motion (2d fBm) as the underlying driving process. The proposed cfGn model is constructed by summing two components, which include periodic deterministic functions modulating the univariate coordinates of 2d fGn. We provide a rigorous derivation of the considered model’s properties, specifically the autocovariance function (ACVF) and frequency-domain characteristics, including the cyclic spectrum. Through theoretical considerations of asymptotic properties and Monte Carlo simulations, we demonstrate that cfGn preserves periodic behavior of ACVF while inheriting long-memory traits which is manifested in time and frequency domains. This framework offers a robust foundation, for instance, in signal-based condition monitoring in systems where periodic fault components coexist with long-range dependent background noise.
\end{abstract}

\end{frontmatter}

\section{Introduction}
This article presents a model integrating two stochastic frameworks: cyclostationarity and long-range dependence (also called long-memory). Thus, let us first discuss such two notions. Cyclostationary processes are defined by periodically time-varying statistics, with particular emphasis on second-order cyclostationarity, where the autocovariance function (ACVF) exhibits periodic behavior. Such models are prevalent in mechanical systems \cite{mech1,antoni2004cyclostationary}, hydrology \cite{hyd1,hyd2}, climatology \cite{met1,met2}, and finance \cite{broszkiewicz2004detecting}. The theoretical groundwork for second-order cyclostationarity was established by Guzdenko and Gladyshev \cite{Guzdenko_1959_Antonio,gladyshev} and subsequently formalized by Hurd \cite{Hurd_1969_Antonio} and Gardner \cite{Gardner_1972a_Antonio}. Since the late 1960s, extensive literature has refined  analytical methods for cyclostationary processes. Methodologically, the analysis of cyclostationary models typically involves examining ACVF in both time and frequency domains, as detailed by Napolitano \cite{napolitano1} and Hurd \cite{hurd2007periodically}. This framework is particularly suited for finite-variance processes and is heavily based on cyclic spectral analysis. A primary diagnostic tool in this field is spectral coherence, which has been extended to applications in condition monitoring and local damage detection \cite{antoni2007cyclic,csc2}.

Various models exhibiting cyclostationarity are well-established in the literature. The most classical examples include periodic autoregressive moving average (PARMA) models, which include periodic coefficients in the ARMA systems  \cite{hurd2007periodically,vecchia1985periodic}, and amplitude-modulated stationary models \cite{gardner2006cyclostationarity}, which are defined as the product of a deterministic periodic function and a wide-sense stationary process. The model proposed in this work represents an extension of the latter class.

The second concept utilized in this article is the long-range dependence property. In practice, this property means that the dependence between observations of a given stochastic process decays extremely slowly as the time distance between them increases. Mathematically, this property is defined through the behavior of the ACVF or the spectral density of a given process. The classical example of the stochastic process with long-range dependence is the fractional Brownian motion (fBm), which was originally introduced by Kolmogorov as "Wiener spirals" \cite{kol140}. The FBm process gained widespread recognition following the seminal work of Mandelbrot and Van Ness \cite{Mandelbrot1968}, who provided its explicit integral representation. FBm is the only Gaussian self-similar process characterized by stationary, power-law correlated increments (called fractional Gaussian noise, fGn) \cite{beran} and its stochastic behavior is uniquely determined by the Hurst exponent. The utility of fBm spans a wide range of scientific disciplines, including hydrology \cite{https://doi.org/10.1029/97WR01982,BENSON2013479}, telecommunications and signal processing \cite{8246429,7448970,984735}, image analysis \cite{605414,6879494}, economics and finance \cite{ROSTEK201330,10.1063/5.0054119,XIAO2010935}, and biological systems, particularly in single-particle tracking experiments \cite{weiss2012,franosch2013,metzler2014,szarek2022statistical,gleb2019}. 

The literature also considers multidimensional versions of fBm, where the individual univariate coordinates are classical fBms that may be linked by a correlation coefficient. The theoretical framework governing multivariate fBm has been investigated extensively; significant contributions, particularly by Amblard et al. \cite{amblard2013basic} and associated studies \cite{coeurjolly2010multivariate, lavancier2009covariance, MAEJIMA1994139}, provide a comprehensive synthesis of the properties of vector-valued, self-similar Gaussian processes with stationary increments. In these contexts, multivariate fBm is generally defined by employing matrix-valued self-similarity exponents alongside matrix-valued covariance structures, where inter-coordinate dependencies are established through constrained cross-covariance parameters. We specifically refer the readers to the recent work of Balcerek et al. \cite{Balcerek2026}, which proposes a construction for a two-dimensional fBm (2d fBm) by introducing a dependence on the underlying noise within the time representation of the process. Balcerek et al. in  \cite{Balcerek2026} discuss also two versions of 2d fBm, well-balanced and causal, which may differ depending on the assumed scenario. 

The model proposed in this article, which we term cyclic fractional Gaussian noise (cfGn), represents an extension of the previously discussed amplitude-modulated stationary process. Specifically, the new model is constructed by integrating two components of an amplitude-modulated framework, each defined as the product of a specific periodic deterministic function and a univariate coordinate of the 2d fGn introduced in \cite{Balcerek2026}. The final process is obtained by summing these two amplitude-modulated components. While the resulting model remains a classical cyclostationary process with a periodically time-varying ACVF, it also possesses the long-memory property for large time lags, a characteristic inherited directly from the underlying 2d fGn. 

Let us note that although this specific configuration mentioned above is novel, the fundamental methodology of constructing similar stochastic processes is well-supported by analogous approaches in the existing literature.  The first appearance of cyclostationary long-memory processes
should be at least dated from the celebrated paper by Hosking \cite{HOSKING1981}, where he considered fractional differentiation of long-memory processes. This construction became known as Gegenbauer process; see review \cite{Dissanayakeetal2018}. The construction of a continuous-time process, which after discretization receives properties similar to Gegenbauer process, is discussed in \cite{anh2004continuous}. Another approach used to design cyclic long-memory processes is to use construction known as random modulation, see, e.g. \cite{papoulis1983random}. Recently, decoupling, modulation, and modeling with such processes have been studied in \cite{kechagias2024cyclical}. Such models are also known as fractional sinusoidal waveform processes \cite{maddanu2022modelling, proietti2024modelling}.

Our contribution with respect to the existing literature discussed above lies in constructing a random modulation process with 2d fGn as the driving process.  In this study, we examine the primary properties of the introduced model expressed through the ACVF. More precisely, we derive the explicit form of the ACVF for both the well-balanced and causal versions of the new model, alongside selected frequency-domain characteristics, such as cyclic frequencies, cyclic autocorrelation functions corresponding to them, and cyclic spectrum. Furthermore, by analyzing the asymptotic behavior of the statistics examined, we demonstrate the unique properties of cfGn. Theoretical considerations are supported by Monte Carlo simulations.

This work serves as a foundational framework for analyzing a new class of models, for instance, in the context of identifying cyclostationary behaviors that coexist with long-memory properties of the signal, or detecting structural changes related to shifts in cyclostationarity in the presence of long-range dependence. Such models and behaviors may manifest themselves for instance in vibration signals from rotating machinery, where cyclostationary patterns typically arise from local faults. A change in the machine's operational state can subsequently alter the cyclostationary characteristics of the signals. In these scenarios, long-range dependence may be associated with additional components that naturally occur in empirical signals.

The rest of the paper is organized as follows. In Section \ref{sec2} we recall the definition of the second-order cyclostationary process along with the definitions of 2d fBm and 2d fGn. In addition, we recall the basic properties of the discussed two-dimensional  processes.  In Section \ref{sec3} we define cfGn and analyze its main properties in time and frequency domain. In Section \ref{sec4}, the asymptotic properties of cfGn are discussed, while in Section \ref{sec5} we verify the theoretical results using Monte Carlo simulations. 
\section{Preliminaries}\label{sec2}
In this section, we recall the definition of a second-order cyclostationary process and present an example process that fulfills this definition. In the second part, we recall the definition of 2d fBm  and its properties. These two notions, i.e.  cyclostationarity and 2d fBm, will then be used to construct a new process that possesses both the cyclostationary property and long-memory behavior.
\subsection{Second-order cyclostationary processes} 
\begin{definition}\cite{hurd2007periodically} {(second-order cyclostationary process)}
The stochastic process  $\{Y(t)\}_{t\ge0}$ is second-order cyclostationary with period $T>0$ if its expected value and autocovariance function (ACVF) are periodic in $t$ with period $T$, i.e., when the following conditions hold for any $s,t \ge0$ 
\begin{eqnarray}
\mathbb{E}[Y(t)]=\mathbb[{E}Y(t+T)],~~ \gamma_Y(s,t)=Cov(Y(s),Y(t))=\gamma_Y(s+T,t+T),
\end{eqnarray}
where $Cov(Y(s),Y(t))=\mathbb{E}[Y(s)Y(t)]-\mathbb{E}[Y(s)]\mathbb{E}[Y(t)]$ is the ACVF of the process $\{Y(t)\}_{t\ge0}$. Here $T$ is the smallest value which satisfies the above property.
\end{definition}
In this article, for simplicity,  the second-order cyclostationary process is called a cyclostationary process. One of the most common examples of the  cyclostationary process is the amplitude-modulated stationary process  defined as follows
\begin{eqnarray}\label{eq1}
    Y(t)=f(t)Z(t),~t\ge0,
\end{eqnarray}
where $f(\cdot)$ is a periodic function with period $T>0$ and $\{Z(t)\}_{t\ge0}$ is the finite-variance stationary process. The classical example of $\{Z(t)\}_{t\ge0}$ process used in Eq. (\ref{eq1}) is the white noise model \cite{gardner2006cyclostationarity}; however, in the literature there are also considered other stationary models used in this definition \cite{napolitano1}. In this article, we extend the classical model defined in Eq. (\ref{eq1}) incorporating possible components of the two-dimensional stationary process, which is discussed in the next part.
\subsection{Two-dimensional fractional Brownian motion and fractional Gaussian noise}\label{2dfBm}
First, we recall definition and important facts related to the 2d fBm and 2d fGn. There are two ways to construct 2d fBm. Let us first introduce the function 
\begin{equation}
    f_{\pm}(x;t,\beta) = (t-x)_{\pm}^{\beta} - (-x)_{\pm}^{\beta},
\end{equation}
where $(x)_+ = \max\{x,0\}$. Let us consider two independent Brownian motions $\{B_1(t)\}_{t\ge0}$, $\{B_2(t)\}_{t\ge0}$ on the real line and define
\begin{equation}
    \begin{cases}
        \tilde{B}_1(t) = B_1(t)\\
        \tilde{B}_2(t) = \rho B_1(t) + \sqrt{1-\rho^2}B_2(t)
    \end{cases}
\end{equation}
Now, we recall the definitions of the so-called "causal" and "well-balanced" 2d fBm.
\begin{definition}\cite{Balcerek2026} {(Causal 2d fBm)}
    Let $H_1,\, H_2\in(0,1)$ and $|\rho|\leq1$. Causal 2d fBm $\{X(t)\}_{t\ge0}$ is defined by 
    \begin{equation}
        X(t) = \int_{\mathbb{R}}\begin{bmatrix}
            \sigma_1 a_{H_1}f_+(s;t,H_1-1/2) & 0\\
            0 & \sigma_2 a_{H_2}f_+(s;t,H_2-1/2)
        \end{bmatrix}
        \begin{bmatrix}
            1 & 0\\
            \rho & \sqrt{1-\rho^2}
        \end{bmatrix}
        \begin{bmatrix}
            dB_1(s)\\
            dB_2(s)
        \end{bmatrix},
    \end{equation}
    where $\sigma_1,\sigma_2>0$. Constants $a_{H_j}, \, j=1,2$ are non-negative and are chosen in such a way that variances of marginals $\{X_j(t)\}_{t\ge0}$ at time $t=1$ are equal to $\sigma_j^2$, i.e.,
    \begin{equation}\label{aconstant}
        a_{H_j}^2 = \frac{\Gamma(2H_j+1)\sin(H_j \pi)}{\Gamma^2(H_j+1/2)}.
    \end{equation}
\end{definition}
\begin{definition}\cite{Balcerek2026} {(Well-balanced 2d fBm)}
    Let $H_1,\, H_2\in(0,1)$ and $|\rho|\leq1$. The Well-balanced 2d fBm $\{X^*(t)\}_{t\ge0}$ is defined as follows 
    \begin{equation}
        X^*(t) = \int_{\mathbb{R}}\begin{bmatrix}
            g_1(s;t) & 0\\
            0 & g_2(s;t)
        \end{bmatrix}
        \begin{bmatrix}
            1 & 0\\
            \rho & \sqrt{1-\rho^2}
        \end{bmatrix}
        \begin{bmatrix}
            dB_1(s)\\
            dB_2(s)
        \end{bmatrix},
    \end{equation}
    where 
    \begin{equation}
        g_j(s;t) = \sigma_ja_{H_j}^*(f_+(s;t,H_j-1/2)+f_-(s:tH_j-1/2))
    \end{equation}
    for $s\in\mathbb{R}$, $t\ge0$, $j=1,2$, $\sigma_1,\sigma_2>0$ and the constants $a_{H_j}^*$ are given by
    \begin{equation}\label{astarconstant}
        a_{H_j}^{*2} = \frac{2H_j(1-2H_j)\pi}{8\Gamma(2-2H_j)\cos(H_j\pi)\Gamma^2(H_j+1/2)\cos^2(\pi(H_j-1/2)/2)},
    \end{equation}
    to ensure that the variances of the marginals $\{X_j(t)\}_{t\ge0}$ at time $t=1$ are equal to $\sigma_j^2$.
\end{definition}
Now, let us recall the covariance structure of 2d fBm. For convenience by $\{Z(t)\}_{t\geq0}$ we denote both causal and well-balanced 2d fBm.
\begin{theorem}\label{th1}\cite{Balcerek2026}
    Let $H_1,\, H_2\in(0,1)$ and $|\rho|\leq1$. The covariance structure of 2d fBm $\{Z(t)\}_{t\geq0}$ is given by
    \begin{equation}
        \gamma_{jk}(s, t) = Cov(Z_j(t), Z_k(s)) = \frac{\sigma_j\sigma_k}{2}(w_{jk}(t)|t|^{H_j+H_k} + w_{jk}(-s)|s|^{H_jH_k} - w_{jk}(t-s)|t-s|^{H_j+H_k}),
    \end{equation}
    for $s,t\ge0$, where $\sigma_j^2 = \mathbb{E}[(Z_j(1))^2]$ and
    \begin{equation}\label{wfunction}
        w_{jk}(u) = \begin{cases}
            \rho_{jk} - \eta_{jk}\text{sign}(u), \quad H_j+H_k\ne1,\\
            \rho_{jk} - \eta_{jk}\text{sign}(u)\log|u|, \quad H_j+H_k\ne1.
        \end{cases}
    \end{equation}
The so called cross correlation parameters $\rho_{12}$ and $\rho_{21}$ are given by
\begin{equation}
    \rho_{12} = \rho_{21} = \rho\frac{\sqrt{\Gamma(2H_1+1)\Gamma(2H_2+1)\sin(H_1\pi)\sin(H_2\pi)}}{\Gamma(H_1+H_2+1)\cos((H_1+H_2)\pi/2)}\cos((H_2-H_1)\pi/2),
\end{equation}
and $\rho_{11}=\rho_{22}=1$ while the assymetry parameters $\eta_{jk}, \, j,k=1,2$, depend on the choice of the model, and in the causal 2d fBm ($Z(t)=X(t)$) they are equal to
\begin{equation}
    \eta_{12}=-\eta_{21}=\rho\frac{\sqrt{\Gamma(2H_1+1)\Gamma(2H_2+1)\sin(H_1\pi)\sin(H_2\pi)}}{\Gamma(H_1+H_2+1)\cos((H_1+H_2)\pi/2)}\sin((H_2-H_1)\pi/2),
\end{equation}
and $\eta_{11}=\eta_{22}=0$, while for well-balanced 2d fBm ($Z(t)=X^*(t)$), we have
\begin{equation}
    \eta_{11}=\eta_{22}=\eta_{12}=\eta_{21}=0.
\end{equation}
\end{theorem}
We define the increment process for time-lag $\tau>0$ by
\begin{equation}
    b_H^{\tau}(t) = Z(t+\tau)-Z(t).
\end{equation}
The increment process of 2d fBm is called 2d fractional Gaussian noise (2d fGn). Similarly to one-dimensional fGn,  2d fGn is stationary process, see \cite{Balcerek2026} for more detiled discussion. When we consider increments for convenience, we will use discrete-time; thus, we consider discrete-time process $\{b_H^{1}(n)\}_{n\in\mathbb{N}_0}$. For the remainder parts of this article, we confine our attention to the discrete-time processes. Now, let us recall the covariance structure of the causal and well-balanced 2d fGn. 
\begin{theorem}\cite{Balcerek2026}
    Let $H_1,\, H_2\in(0,1)$ and $|\rho|\leq1$. The covariance structure of the causal or well-balanced 2d fGn  $\{b_H^{1}(n)\}_{n\in\mathbb{N}_0}$is given by
    \begin{equation}\label{gammaeq}
        \gamma_{jk}^{1}(h) = Cov(b_H^{1}(n),b_H^{1}(n+h)) = \frac{\sigma_j\sigma_k\rho_{jk}}{2}[w_{jk}(h+1)(h+1)^{H_j+H_k} + w_{jk}(h-1)|h-1|^{H_j+H_k} - 2w_{jk}(h)h^{H_j+H_k}],
    \end{equation}
    where the function $w_{jk}$ is given in Eq. \eqref{wfunction} and depends on whether we consider causal or
well-balanced 2d fBm. The coefficients $\rho_{jk}$ and $\eta_{jk}$ are defined in Theorem \ref{th1}.
\end{theorem}

Beyond characterizing the time-domain properties of 2d fBm, it is essential to establish its behavior within the frequency domain. To this end, we first review the spectral density of 2d fGn.

\begin{definition}\cite{brockwell1991time}
    Let us consider a 2d discrete-time stationary stochastic process with zero-mean and covariance matrix $\gamma^1(h)=[\gamma_{ij}^1(h)]_{i,j=1}^{2}$, $h\in \mathbb{Z}$. If there exists corresponding spectral density matrix, then it  is given by
\begin{equation}
    f^1(\lambda) = [f_{ij}^1(h)]_{i,j=1}^{2} = \sum_{h=-\infty}^{\infty}e^{ih\lambda}\gamma^{1}(h) = \left[\sum_{h=-\infty}^{\infty}e^{ih\lambda}\gamma_{ij}^1(h)\right]_{i,j=1}^{2}, \quad \lambda\in\mathbb{R}.
\end{equation}
Now, we recall the form of the spectral density matrix of 2d fGn and its asymptotic properties.
\end{definition}
\begin{theorem}\cite{Balcerek2026}\label{th3}
Let $H_1, H_2 \in (0, 1)$, $H = \text{diag}(H_1, H_2)$ be a diagonal matrix with elements $H_1, H_2$, and $|\rho| \le 1$. Additionally, let matrices $C_c = [c^c_{jk}]_{j,k=1,2}$ and $C_{wb} = [c^{wb}_{jk}]_{j,k=1,2}$ that correspond to causal and well-balanced 2d fBm, respectively, have the following elements
\begin{align}
    c^c_{jj} &= \frac{1}{2\pi} \sigma_j^2 \Gamma^2(H_j + 0.5) a_{H_j}^2, \\
    c^c_{jk} &= \frac{1}{2\pi} \rho \sigma_j \sigma_k \Gamma(H_j + 0.5) \Gamma(H_k + 0.5) a_{H_j} a_{H_k} e^{-i \frac{\pi}{2}(H_j - H_k)}, \quad j \neq k, \\
    c^{wb}_{jj} &= \frac{2}{\pi} \sigma_j^2 \cos^2\left(\frac{(H_j - 0.5)\pi}{2}\right) \Gamma^2(H_j + 0.5) a_{H_j}^{*2}, \\
    c^{wb}_{jk} &= \frac{1}{2\pi} \rho \sigma_j \sigma_k \cos\left(\frac{(H_j - 0.5)\pi}{2}\right) \cos\left(\frac{(H_k - 0.5)\pi}{2}\right) \Gamma(H_j + 0.5) \Gamma(H_k + 0.5) a_{H_j}^* a_{H_k}^*, \quad j \neq k,
\end{align}
for $j, k = 1, 2$. The constants $a_{H_j}$ and $a_{H_j}^*$ are given in Eqs. \eqref{aconstant} and \eqref{astarconstant}, respectively. Then, for 2d fGn, the spectral density matrix $f^1(\lambda)$ is given by
\begin{equation}
    f^1(\lambda) = |1 - e^{-i\lambda}|^2 \sum_{n=-\infty}^{\infty} \frac{(\lambda + 2\pi n)^{-D} \tilde{C} (\lambda + 2\pi n)^{-D}_+ + (\lambda + 2\pi n)^{-D} \bar{\tilde{C}} (\lambda + 2\pi n)^{-D}_-}{(\lambda + 2\pi n)^2},
\end{equation}
where $(x)_+ \equiv \max\{x, 0\}$, $(x)_- \equiv \max\{-x, 0\}$,
\begin{equation}\label{matd}
D = [d_{ij}]_{i,j=1}^2 = \text{diag}(H_1, H_2) - 0.5I_2,
\end{equation}
$I_2$ is a $2 \times 2$ an identity matrix, and $\bar{C}$ denotes the element-wise complex conjugate of matrix $C$. The matrix 
\begin{equation}\label{matc}
    \tilde{C} = [\tilde{c}_{ij}]_{i,j=1}^2
\end{equation}
is equal to $C_c$ or $C_{wb}$ depending on whether we consider a causal or well-balanced 2d fBm. The components $ f^1_{jk}(\lambda)$ can thus be expressed as
\begin{equation}\label{fij1}
    f^1_{jk}(\lambda) = |1 - e^{-i\lambda}|^2 \sum_{n=-\infty}^{\infty} \frac{(\lambda + 2\pi n)_+^{1-H_j-H_k} \tilde{c}_{jk} + (\lambda + 2\pi n)_-^{1-H_j-H_k} \bar{\tilde{c}}_{jk}}{(\lambda + 2\pi n)^2},
\end{equation}
where $\tilde{c}_{jk}$ are elements of the appropriate matrix. Moreover, for $\lambda \to 0$ we have
\begin{equation}
    f^1(\lambda) \sim \lambda^{-D} \tilde{C} \lambda^{-D}.
\end{equation}
The element-wise asymptotic is
\begin{equation}
    f^1_{jk}(\lambda) \sim \tilde{c}_{jk} \lambda^{-(d_j + d_k)} = \tilde{c}_{jk} \lambda^{-(H_j + H_k - 1)},
\end{equation}
where $d_j$, $d_k$ are diagonal elements of matrix $D$.
\end{theorem}
More details on 2d fBm and 2d fGn can be found in \cite{Balcerek2026}.
\section{Cyclic fractional Gaussian noise and its properties}\label{sec3}
In this section, we introduce the cyclic fractional Gaussian noise (cfGn), which is an extension of the classical cyclostationary process defined in Eq. (\ref{eq1}). 
\begin{definition}{(cyclic fractional Gaussian noise)} The cfGn is defined as follows
\begin{eqnarray}\label{eq2}
    Y(n) \eqdef a_1\cos(\lambda_0n)b_{1,H}^{1}(n)+a_2\sin(\lambda_0n)b_{2,H}^{1}(n),~~n\in \mathbb{N}_0,
\end{eqnarray}
where $\{b_{1,H}^{1}(n)\}_{n\in\mathbb{N}_0}$, $\{b_{2,H}^{1}(n)\}_{n\in\mathbb{N}_0}$ are components of two-dimensional fGn defined in Section \ref{2dfBm} and $\lambda_0,a_1,a_1\in \mathbb{R}$.
\end{definition}
Note that cfGn is a discrete-time Gaussian process with zero mean. In addition, when $a_1$ or $a_2$ are equal to zero, this process satisfies Eq. (\ref{eq1}). Thus, it can be considered as an extension of the classical cyclostationary model. On the other hand, it can also be considered as the extension of the process defined in \cite{kechagias2024cyclical}, where the authors discussed the model defined as in (\ref{eq2}) with the components $\{b_{1,H}(n)\}_{n\in\mathbb{N}_0}$, $\{b_{2,H}(n)\}_{n\in\mathbb{N}_0}$ related by the specific covariance matrix. In our analysis, for simplicity, we assume $a_1=a_2=1$; however, the presented methodologies are also valuable in the general case. In the next part, we discuss the main properties of the model (\ref{eq2}) expressed in terms of its ACVF in the time and frequency domains. First, we calculate the ACVF of $\{Y(n)\}_{n\in\mathbb{N}_0}$.
\begin{proposition}\label{acf}
    The ACVF of cfGn $\{Y(n)\}_{n\in\mathbb{N}_0}$ is given by
    \begin{equation}\label{yautocov}
        \gamma_Y(n, h) = \dot{\gamma}(h) + \tilde{\gamma}(n, h),
    \end{equation}
    where \begin{align}
    \dot{\gamma}(h) &= \frac{1}{2} \left[ (\gamma_{11}^1(h) + \gamma_{22}^{1}(h))\cos(\lambda_0 h) + (\gamma_{21}^{1}(h) - \gamma_{12}^{1}(h))\sin(\lambda_0 h) \right], \\
    \tilde{\gamma}(n, h) &= \frac{1}{2} \left[ (\gamma_{11}^{1}(h) - \gamma_{22}^{1}(h))\cos(\lambda_0(2n+h)) + (\gamma_{21}^{1}(h) + \gamma_{12}^{1}(h))\sin(\lambda_0(2n+h)) \right],
\end{align}
and $\gamma_{ij}^1$ are defined in Eq. \eqref{gammaeq}.
\begin{proof}
 Since $\{Y(n)\}_{n\in\mathbb{N}_0}$ is a zero-mean process, we have
    \begin{equation}\label{varproof1}
    \begin{aligned}
    \gamma_Y(n, h) = \mathbb{E}[Y(n+h)Y(n)] &= \mathbb{E}\left[ (\cos(\lambda_0(n+h))b_{1,n+h}^1 + \sin(\lambda_0(n+h))b_{2,n+h}^1)(\cos(\lambda_0 n)b_{1,n}^1 + \sin(\lambda_0 n)b_{2,n}^1) \right] \\
    &= \cos(\lambda_0(n+h))\cos(\lambda_0 n)\gamma_{11}(h) + \sin(\lambda_0(n+h))\sin(\lambda_0 n)\gamma_{22}(h) \\
    &+ \cos(\lambda_0(n+h))\sin(\lambda_0 n)\gamma_{12}(h) + \sin(\lambda_0(n+h))\cos(\lambda_0 n)\gamma_{21}(h).
\end{aligned}
\end{equation}
Using trigonometric identities
$\cos(\alpha) \cos(\beta) = (\cos(\alpha-\beta) + \cos(\alpha+\beta))/2$ and $\sin(\alpha) \sin (\beta) = (\cos(\alpha-\beta) - \cos(\alpha+\beta))/2$,
$\sin(\alpha) \cos(\beta) = (\sin(\alpha+\beta) + \sin(\alpha-\beta))/2$ and $\cos(\alpha) \sin(\beta) = (\sin(\alpha+\beta) - \sin(\alpha-\beta))/2$, plugging them into Eq. \eqref{varproof1} and grouping of terms with $n$ and without $n$ yields the thesis.
\end{proof}
\end{proposition}

We now proceed to characterize the  properties of the process $\{Y(n)\}_{n\in\mathbb{N}_0}$ in frequency domain. We begin by revisiting the fundamental definitions of cyclic frequencies and corresponding cyclic autocorrelation functions (CAFs).
\begin{definition}
    Assume that the ACVF of a second-order stationary stochastic process has a convergent Fourier series expansion \cite{napolitano1}
    \begin{equation} 
        \gamma_Y(n, h) = \sum_{\alpha \in \mathcal{A}} R_Y^\alpha(h) e^{i\alpha n},
    \end{equation}
    then, $R_Y^\alpha(h)$ are called CAFs and $\mathcal{A}$ is called a set of non-zero cyclic frequencies.  
\end{definition}  
First, we identify the set of non-zero cyclic frequencies $\mathcal{A}$ and calculate CAFs of cfGn.
\begin{proposition}\label{cafprop}
    The set of non-zero cyclic frequencies of the cfGn $\{Y(n)\}_{n\in\mathbb{N}_0}$ is $\mathcal{A} = \{0, 2\lambda_0, -2\lambda_0\}$. The corresponding CAFs (i.e. $R_Y^\alpha(h)$), are given by
    \begin{align}
        R_Y^0(h) &= \dot{\gamma}(h), \\
        R_Y^{2\lambda_0}(h) &= \frac{e^{i\lambda_0 h}}{4} \left[ (\gamma_{11}^{1}(h) - \gamma_{22}^{1}(h)) - i(\gamma_{21}^{1}(h) + \gamma_{12}^{1}(h)) \right], \\
        R_Y^{-2\lambda_0}(h) &= \overline{R_Y^{2\lambda_0}(h)},
    \end{align}
    where functions $\gamma^1_{ij}(h)$ are defined in Eq. \eqref{gammaeq}.
\end{proposition}

\begin{proof}
From Proposition \ref{acf} we have that the ACVF of the process $\{Y(n)\}_{n\in\mathbb{N}_0}$ is decomposed as $\gamma_Y(n, h) = \dot{\gamma}(h) + \tilde{\gamma}(n, h)$.
    The term $\dot{\gamma}(h)$ is independent of $n$, i.e., $R_Y^0(h) = \dot{\gamma}(h)$. Next, we analyze the n-dependent component $\tilde{\gamma}(n, h)$. We have the following
    \begin{equation}\label{autocovndepaux}
        \tilde{\gamma}(n, h) = \frac{\gamma_{11}^{1}(h) - \gamma_{22}^{1}(h)}{2}\cos(\lambda_0(2n+h)) + \frac{\gamma_{21}^{1}(h) + \gamma_{12}^{1}(h)}{2}\sin(\lambda_0(2n+h)).
    \end{equation}
    Using Euler's formulas $\cos(\theta) = \frac{e^{i\theta} + e^{-i\theta}}{2}$ and $\sin(\theta) = \frac{e^{i\theta} - e^{-i\theta}}{2i}$, we expand the trigonometric terms with $\theta = 2\lambda_0 n + \lambda_0 h$
    \begin{align}
        \cos(2\lambda_0 n + \lambda_0 h) &= \frac{1}{2}e^{i\lambda_0 h}e^{i2\lambda_0 n} + \frac{1}{2}e^{-i\lambda_0 h}e^{-i2\lambda_0 n}, \\
        \sin(2\lambda_0 n + \lambda_0 h) &= \frac{1}{2i}e^{i\lambda_0 h}e^{i2\lambda_0 n} - \frac{1}{2i}e^{-i\lambda_0 h}e^{-i2\lambda_0 n}.
    \end{align}
    Substituting these into Eq. \eqref{autocovndepaux} and grouping with respect to the exponential functions $e^{i2\lambda_0 n}$ and $e^{-i2\lambda_0 n}$, we obtain the following
    \begin{align}
        \tilde{\gamma}(n, h) &= e^{i2\lambda_0 n} \left[ \frac{\gamma_{11}^{1}(h) - \gamma_{22}^{1}(h)}{4}e^{i\lambda_0 h} + \frac{\gamma_{21}^{1}(h) + \gamma_{12}^{1}(h)}{4i}e^{i\lambda_0 h} \right] \nonumber \\
        &+ e^{-i2\lambda_0 n} \left[ \frac{\gamma_{11}^{1}(h) - \gamma_{22}^{1}(h)}{4}e^{-i\lambda_0 h} - \frac{\gamma_{21}^{1}(h) + \gamma_{12}^{1}(h)}{4i}e^{-i\lambda_0 h} \right].
    \end{align}
    Using the identity $1/i = -i$, the expression multiplying $e^{i2\lambda_0 n}$ becomes
    \begin{equation}
        R_Y^{2\lambda_0}(h) = \frac{e^{i\lambda_0 h}}{4} \left[ (\gamma_{11}^{1}(h) - \gamma_{22}^{1}(h)) - i(\gamma_{21}^{1}(h) + \gamma_{12}^{1}(h)) \right].
    \end{equation}
    Similarly, the expression multiplying $e^{-i2\lambda_0 n}$ is
    \begin{equation}
        R_Y^{-2\lambda_0}(h) = \frac{e^{-i\lambda_0 h}}{4} \left[ (\gamma_{11}^{1}(h) - \gamma_{22}^{1}(h)) + i(\gamma_{21}^{1}(h) + \gamma_{12}^{1}(h)) \right].
    \end{equation}
    We observe that $R_Y^{-2\lambda_0}(h)$ is indeed the complex conjugate of $R_Y^{2\lambda_0}(h)$. Since the expansion fully accounts for all terms in $\gamma_Y(n, h)$, the coefficients for all other $\alpha$ values are zero.
\end{proof}
Now, let us recall the definition of the cyclic spectrum. In the following analysis, by $\mathcal{F}\{f\}(\lambda)$ we denote the Fourier transform of the function $f$.
\begin{definition}
    Let $R_Y^\alpha(h)$ be the CAF corresponding to its non-zero cyclic frequency 
    $\alpha \in \mathcal{A}$ for some discrete-time stationary stochastic process with zero-mean. Then, the cyclic spectrum at frequency $\alpha$ is given by \cite{napolitano1} 
    \begin{equation}
        S_Y^\alpha(\lambda) = \sum_{h=-\infty}^{\infty} R_Y^\alpha(h) e^{-i\lambda h}.
    \end{equation}
\end{definition}Next, we calculate the cyclic spectrum of the cfGn.
\begin{proposition}
    The cyclic spectrum of cfGn $\{Y(n)\}_{n\in\mathbb{N}_0}$ for the cyclic frequency $\alpha = 0$ is given by
    \begin{equation}\label{cs0}
        S_Y^0(\lambda) = \frac{1}{4} \sum_{s \in \{-1, 1\}} \left[ f_{11}^{1}(\lambda - s\lambda_0) + f_{22}^{1}(\lambda - s\lambda_0) - 2s \Im f_{12}^{1}(\lambda - s\lambda_0) \right].
    \end{equation}
    For the cyclic frequency $\alpha = 2\lambda_0$ the corresponding cyclic spectrum is given by
    \begin{equation}\label{cs2l0}
        S_Y^{2\lambda_0}(\lambda) = \frac{1}{4} \left[ f_{11}^{1}(\lambda - \lambda_0) - f_{22}^{1}(\lambda - \lambda_0) - 2i \Re f_{12}^{1}(\lambda - \lambda_0) \right].
    \end{equation}
    For the cyclic frequency $\alpha = -2\lambda_0$ the corresponding cyclic spectrum is given by
    \begin{equation}
        S_Y^{-2\lambda_0}(\lambda) = \overline{S_Y^{2\lambda_0}(-\lambda)} = \frac{1}{4} \left[ f_{11}^{1}(\lambda + \lambda_0) - f_{22}^{1}(\lambda + \lambda_0) + 2i \Re f_{12}^{1}(\lambda + \lambda_0) \right].
    \end{equation}
In the above formulas, $\Re$ and $\Im$ are real and imaginary parts of the complex number, respectively, and the functions $f_{ij}^1(\lambda)$ are given in Eq. \eqref{fij1}.
\end{proposition}

\begin{proof}
To prove the proposition, we use the linearity of the Fourier transform and the property $\mathcal{F}\{f(h) e^{i\omega h}\}(\lambda) = \mathcal{F}\{f(h)\} (\lambda - \omega)$.
    
    \paragraph{1. Case $\alpha = 0$}
 Recall that \begin{equation}
    R_Y^0(h) = \frac{1}{2} \underbrace{(\gamma_{11}^{1}(h) + \gamma_{22}^{1}(h))}_{A(h)} \cos(\lambda_0 h) + \frac{1}{2} \underbrace{(\gamma_{21}^{1}(h) - \gamma_{12}^{1}(h))}_{B(h)} \sin(\lambda_0 h).
\end{equation}
 Using the fact that $\mathcal{F}\{A(h)\cos(\lambda_0 h)\}(\lambda) = \frac{1}{2}[F_A(\lambda-\lambda_0) + F_A(\lambda+\lambda_0)]$, where $F_A(\lambda) = f_{11}^{1}(\lambda) + f_{22}^{1}(\lambda)$, we have the following
    \begin{equation}
        \mathcal{F}\left\{ \frac{A(h)}{2}\cos(\lambda_0 h) \right\}(\lambda) = \frac{1}{4} \left[ f_{11}^{1}(\lambda - \lambda_0) + f_{22}^{1}(\lambda - \lambda_0)  + f_{11}^{1}(\lambda + \lambda_0) + f_{22}^{1}(\lambda + \lambda_0)\right].
    \end{equation}
    Using $\mathcal{F}\{B(h)\sin(\lambda_0 h)\} = \frac{1}{2i}[F_B(\lambda-\lambda_0) - F_B(\lambda+\lambda_0)]$, where
    $F_B(\lambda) = f_{21}^{1}(\lambda) - f_{12}^{1}(\lambda) = \overline{f_{12}^{1}(\lambda)} - f_{12}^{1}(\lambda) = -2i \Im f_{12}^{1}(\lambda)$, we obtain
    \begin{align}
        \mathcal{F}\left\{ \frac{B(h)}{2}\sin(\lambda_0 h) \right\}(\lambda) &= \frac{1}{4i} \left[ (-2i \Im f_{12}^{1}(\lambda-\lambda_0)) - (-2i \Im f_{12}^{1}(\lambda+\lambda_0)) \right] \nonumber \\
        &= -\frac{1}{2} \Im f_{12}^{1}(\lambda-\lambda_0) + \frac{1}{2} \Im f_{12}^{1}(\lambda+\lambda_0).
    \end{align}
Combining these results, we obtain the final expression in the following form
\begin{equation}
    S_Y^0(\lambda) = \frac{1}{4} \sum_{s \in \{-1, 1\}} \left[ f_{11}^{1}(\lambda - s\lambda_0) + f_{22}^{1}(\lambda - s\lambda_0) - 2s \Im f_{12}^{1}(\lambda - s\lambda_0) \right].
\end{equation}

    \paragraph{2. Case $\alpha = 2\lambda_0$}
    The CAF at this frequency is given by
    \begin{equation}
        R_Y^{2\lambda_0}(h) = \frac{e^{i\lambda_0 h}}{4} \left[ (\gamma_{11}^{1}(h) - \gamma_{22}^{1}(h)) - i(\gamma_{21}^{1}(h) + \gamma_{12}^{1}(h)) \right].
    \end{equation}
    Taking the sum over $h\in\mathbb{Z}$ in the expression $\sum_{h \in \mathbb{Z}} R_Y^{2\lambda_0}(h) e^{-i\lambda h}$, we receive the following
    \begin{align}\label{s2l0first}
        S_Y^{2\lambda_0}(\lambda) &= \frac{1}{4} \sum_{h=-\infty}^{\infty} e^{-i(\lambda - \lambda_0)h} \left[ \gamma_{11}^{1}(h) - \gamma_{22}^{1}(h) - i(\gamma_{21}^{1}(h) + \gamma_{12}^{1}(h)) \right] \nonumber \\
        &= \frac{1}{4} \left[ f_{11}^{1}(\lambda - \lambda_0) - f_{22}^{1}(\lambda - \lambda_0) - i(f_{21}^{1}(\lambda - \lambda_0) + f_{12}^{1}(\lambda - \lambda_0)) \right].
    \end{align}
    Using the fact that $f_{21}^{1}(\lambda) = \overline{f_{12}^{1}(\lambda)}$, we have $f_{21}^{1}(\lambda) + f_{12}^{1}(\lambda) = 2\Re f_{12}^{1}(\lambda)$. Substituting this into Eq. \eqref{s2l0first} yields
    \begin{equation}
        S_Y^{2\lambda_0}(\lambda) = \frac{1}{4} \left[ f_{11}^{1}(\lambda - \lambda_0) - f_{22}^{1}(\lambda - \lambda_0) - 2i \Re f_{12}^{1}(\lambda - \lambda_0) \right].
    \end{equation}
    \paragraph{3. Case $\alpha=-2\lambda_0$}
 Since $R_Y^{-2\lambda_0}(h)=R_Y^{2\lambda_0}(h)$, we immediately have $S_Y^{-2\lambda_0}(\lambda) = \overline{S_Y^{2\lambda_0}(-\lambda)}$.
\end{proof}
\section{Asymptotic properties of CfGn}\label{sec4}
In this section, we calculate the asymptotic behavior of ACVF and the cyclic spectrum of the process $\{Y(n)\}_{n\in\mathbb{N}_0}$. Let us note that asymptotic properties are independent of using causal or well-balanced 2d fBm as an underlying process defined in Eq. (\ref{eq2}).
\begin{proposition}\label{autocovasympt}
    As the lag $h \to \infty$, the ACVF $\gamma_Y(n, h)$ of cfGn $\{Y(n)\}_{n\in\mathbb{N}_0}$ behaves as
    \begin{equation}
    \gamma_Y(n, h) \sim
        \begin{cases}
            \frac{c_{11}}{2} h^{2H_1-2} \left[ \cos(\lambda_0 h) + \cos(\lambda_0(2n+h)) \right] \text{ for } H_1 > H_2,\\
            \frac{c_{22}}{2} h^{2H_2-2} \left[ \cos(\lambda_0 h) - \cos(\lambda_0(2n+h)) \right] \text{ for } H_2 > H_1,\\
            h^{2H-2} \left[ A_{\text{stat}} \cos(\lambda_0 h - \phi_{\text{stat}}) + A_{\text{cyc}} \cos(\lambda_0(2n+h) - \phi_{\text{cyc}}) \right] \text{ for } H_1 = H_2 = H.
        \end{cases}
    \end{equation}


\noindent     where $c_{jk} = \sigma_j \sigma_k \rho_{jk} \frac{(H_j+H_k)(H_j+H_k-1)}{2}$, $A_{\text{stat}} = \frac{1}{2} \sqrt{ (c_{11}+c_{22})^2 + (c_{21}-c_{12})^2 }$, $\phi_{\text{stat}} = \arctan\left( \frac{c_{21}-c_{12}}{c_{11}+c_{22}} \right)$, \newline $A_{\text{cyc}} = \frac{1}{2} \sqrt{ (c_{11}-c_{22})^2 + (c_{21}+c_{12})^2 }$, $\phi_{\text{cyc}} = \arctan\left( \frac{c_{21}+c_{12}}{c_{11}-c_{22}} \right)$.
\end{proposition}

\begin{proof}
    Let us note that the ACVF of 2d fGn components behaves asymptotically as $\gamma_{jk}^{1}(h) \sim c_{jk} h^{H_j+H_k-2}$. Now, based on Eq. \eqref{yautocov}, we analyze three cases depending on the values of $ H_1$ and $H_2$.
    \paragraph{1. Case $H_1 > H_2$}
    In this case, the term $\gamma_{11}^{1}(h) \sim h^{2H_1-2}$ decays the slowest, and all other terms involving $H_2$ are negligible as $h \to \infty$. Thus, we have
    \begin{align}
        \gamma_Y(n, h) &\sim \cos(\lambda_0(n+h))\cos(\lambda_0 n) \gamma_{11}^{1}(h) \nonumber \sim \frac{1}{2}[\cos(\lambda_0(2n+h)) + \cos(\lambda_0 h)] c_{11} h^{2H_1-2}.
    \end{align}

    \paragraph{2. Case $H_2 > H_1$}
    In this case, the term $\gamma_{22}^{1}(h) \sim h^{2H_2-2}$ decays the slowest, and all other terms involving $H_1$ are negligible as $h \to \infty$. Thus, we have
    \begin{align}
        \gamma_Y(n, h) &\sim \sin(\lambda_0(n+h))\sin(\lambda_0 n) \gamma_{22}^{1}(h) \sim \frac{1}{2}[\cos(\lambda_0 h) - \cos(\lambda_0(2n+h))] c_{22} h^{2H_2-2}.
    \end{align}

    \paragraph{3. Case $H_1 = H_2 = H$}
    In this case, all components decay at the same rate $h^{2H-2}$. We analyze $\dot{\gamma}(h)$ and $\tilde{\gamma}(n, h)$ separately. For $\dot{\gamma}(h)$, we have
    \begin{equation}
        \dot{\gamma}(h) \sim \frac{h^{2H-2}}{2} \left[ (c_{11}+c_{22})\cos(\lambda_0 h) + (c_{21}-c_{12})\sin(\lambda_0 h) \right].
    \end{equation}
    Using the identity $r_1\cos(\theta) + r_2\sin(\theta) = \sqrt{r_1^2+r_2^2}\cos(\theta - \arctan(r_2/r_1))$, with $r_1 = c_{11}+c_{22}$ and $r_2 = c_{21}-c_{12}$, we obtain $A_{\text{stat}}$ and $\phi_{\text{stat}}$. For $\tilde{\gamma}(n, h)$, we have
    \begin{equation}
        \tilde{\gamma}(n, h) \sim \frac{h^{2H-2}}{2} \left[ (c_{11}-c_{22})\cos(\lambda_0(2n+h)) + (c_{21}+c_{12})\sin(\lambda_0(2n+h)) \right].
    \end{equation}
    Similarly, setting $r_1 = c_{11}-c_{22}$ and $r_2 = c_{21}+c_{12}$ yields $A_{\text{cyc}}$ and $\phi_{\text{cyc}}$.
\end{proof}
Next, we analyze the asymptotics of the cyclic spectrum of cfGn at the frequency $2\lambda_0$ given in Eq. \eqref{cs2l0}.
\begin{proposition}\label{csasympt}
Let $d_{\max} = \max(d_1, d_2)$, where $d_1$, $d_2$ are diagonal elements of matrix $D$ given in Eq. \eqref{matd}. Then, the cyclic spectrum of cfGn at cyclic frequency $2\lambda_0$ given in Eq. \eqref{cs2l0} behaves like
    \begin{equation}
        S_Y^{2\lambda_0}(\lambda) \sim \mathcal{K} |\lambda - \lambda_0|^{-2d_{\max}}, \quad \text{as } \lambda \to \lambda_0,
    \end{equation}
    where the coefficient $\mathcal{K}$ is given by
    \begin{equation}
        \mathcal{K} = \begin{cases}
            \frac{\tilde{c}_{11}}{4} \text{ for } d_1>d_2,\\
            -\frac{\tilde{c}_{22}}{4} \text{ for } d_2>d_1,\\
            \frac{1}{4} \left[ \tilde{c}_{11} - \tilde{c}_{22} - 2i \Re \tilde{c}_{12} \right] \text{ for } d_1=d_2=d.
        \end{cases}
    \end{equation}
    where $\tilde{c}_{ij}$, $i,j=1,2$ are given in Eq. \eqref{matc}.
\end{proposition}

\begin{proof}
    We start with the exact expression for the cyclic spectrum at the cycle frequency $\alpha = 2\lambda_0$
    \begin{equation}
        S_Y^{2\lambda_0}(\lambda) = \frac{1}{4} \left[ f_{11}^{1}(\lambda - \lambda_0) - f_{22}^{1}(\lambda - \lambda_0) - 2i \Re f_{12}^{1}(\lambda - \lambda_0) \right].
    \end{equation}
    
    Let $\omega = \lambda - \lambda_0$. We analyze the behavior of $S_Y^{2\lambda_0}(\lambda)$ as $\omega \to 0$ substituting the given element-wise asymptotics $f_{jk}(\omega) \sim \tilde{c}_{jk} |\omega|^{-(d_j + d_k)}$.
    
    \text{1. Case $d_1 > d_2$.}
    In this case, we have $2d_1 > d_1 + d_2 > 2d_2$.
    Therefore, $|\omega|^{-2d_1}$ is the dominant term and we have
    \begin{align}
        S_Y^{2\lambda_0}(\lambda) &= \frac{1}{4} \left[ \tilde{c}_{11}|\omega|^{-2d_1} - O(|\omega|^{-2d_2}) - O(|\omega|^{-(d_1+d_2)}) \right] \sim \frac{\tilde{c}_{11}}{4} |\lambda - \lambda_0|^{-2d_1}.
    \end{align}

    \text{2. Case $d_2 > d_1$.}
    Here $2d_2$ is the largest exponent. The term $-f_{22}$ dominates, and thus we have
    \begin{align}
        S_Y^{2\lambda_0}(\lambda) &= \frac{1}{4} \left[ O(|\omega|^{-2d_1}) - \tilde{c}_{22}|\omega|^{-2d_2} - O(|\omega|^{-(d_1+d_2)}) \right] \sim -\frac{\tilde{c}_{22}}{4} |\lambda - \lambda_0|^{-2d_2}.
    \end{align}

    \text{3. Case $d_1 = d_2 = d$.}
    All terms diverge at the same rate $|\omega|^{-2d}$, and thus we obtain
    \begin{align}
        S_Y^{2\lambda_0}(\lambda) &\sim \frac{1}{4} \left[ \tilde{c}_{11}|\omega|^{-2d} - \tilde{c}_{22}|\omega|^{-2d} - 2i \operatorname{Re}(\tilde{c}_{12} |\omega|^{-2d}) \right] = \frac{|\lambda - \lambda_0|^{-2d}}{4} \left[ \tilde{c}_{11} - \tilde{c}_{22} - 2i \operatorname{Re}(\tilde{c}_{12}) \right].
    \end{align}
\end{proof}
\section{Numerical simulations}\label{sec5}
In this section, the theoretical results presented in the previous sections are compared with the results obtained through Monte Carlo simulations. For all statistics, we used 10000 Monte Carlo simulations and applied well-balanced and causal 2d fGn as the underlying process. For all simulations, we fix the parameters $\lambda_0=0.1\pi$ and $\rho\in\{-0.15, 0, 0.15\}$. In Figs. \ref{fig::trajectories}-\ref{fig::csim} we set $H_1=0.4$, $H_2=0.7$. For the asymptotic results shown in Figs. \ref{fig::covasympt} and \ref{fig::csasympt}, we set the parameters to $(H_1, H_2)\in\{(0.85, 0.4), (0.4, 0.85), (0.75, 0.75)\}$.

First, in Fig. \ref{fig::trajectories} we present example trajectories of cfGn $\{Y(n)\}_{n\in\mathbb{N}_0}$. Fig.~\ref{fig::autocavariance} illustrates a comparison between the theoretical ACVF of the process $\{Y(n)\}_{n\in\mathbb{N}_0}$, as defined in Eq. \eqref{yautocov}, and the empirical one derived from Monte Carlo simulations. Across all examined cases, the analytical expressions demonstrate excellent agreement with the simulated trajectories.

A similar comparative analysis is presented in Figs. \ref{fig::cafre} and \ref{fig::cafim}, which depict the real and imaginary parts of the CAF at frequency $2\lambda_0$, respectively.  Also, for this quantity, the agreement between the theoretical and simulation results is clearly visible. 

The spectral characteristics are further evaluated in Figs. \ref{fig::cs0} through \ref{fig::csim}. Specifically, Fig. \ref{fig::cs0} validates the theoretical cyclic spectrum at frequency $0$ against empirical estimates. Figs. \ref{fig::csre} and \ref{fig::csim} provide the real and imaginary parts of the cyclic spectrum at frequency $2\lambda_0$, see Eq. \eqref{cs2l0}. While the theoretical and empirical results align closely for $\rho \in\{-0.15,0.15\}$, a slight divergence is observed in the imaginary component when $\rho=0$. However, since the theoretical value in this instance is null, the observed discrepancy is attributed to inherent stochastic approximation errors (simulation noise) rather than a model mismatch.

Finally, the asymptotic behavior of the cfGn is examined. Fig. \ref{fig::covasympt} shows that the theoretical asymptotics of the ACVF from Proposition \ref{autocovasympt} align precisely with the simulated results. In Fig. \ref{fig::csasympt}, the absolute value of the CAF asymptotics at $2\lambda_0$ (Proposition \ref{csasympt}) is compared with empirical counterparts. Consistent with previous observations, the theoretical framework holds for all cases except when $\rho=0$ and $H_1=H_2$. In this specific configuration, the theoretical value is zero, and the minor deviations observed in the simulations are representative of numerical noise limits.

\begin{figure}[ht!]
       \centering
         \includegraphics[width=1\textwidth, height=0.5\textheight]{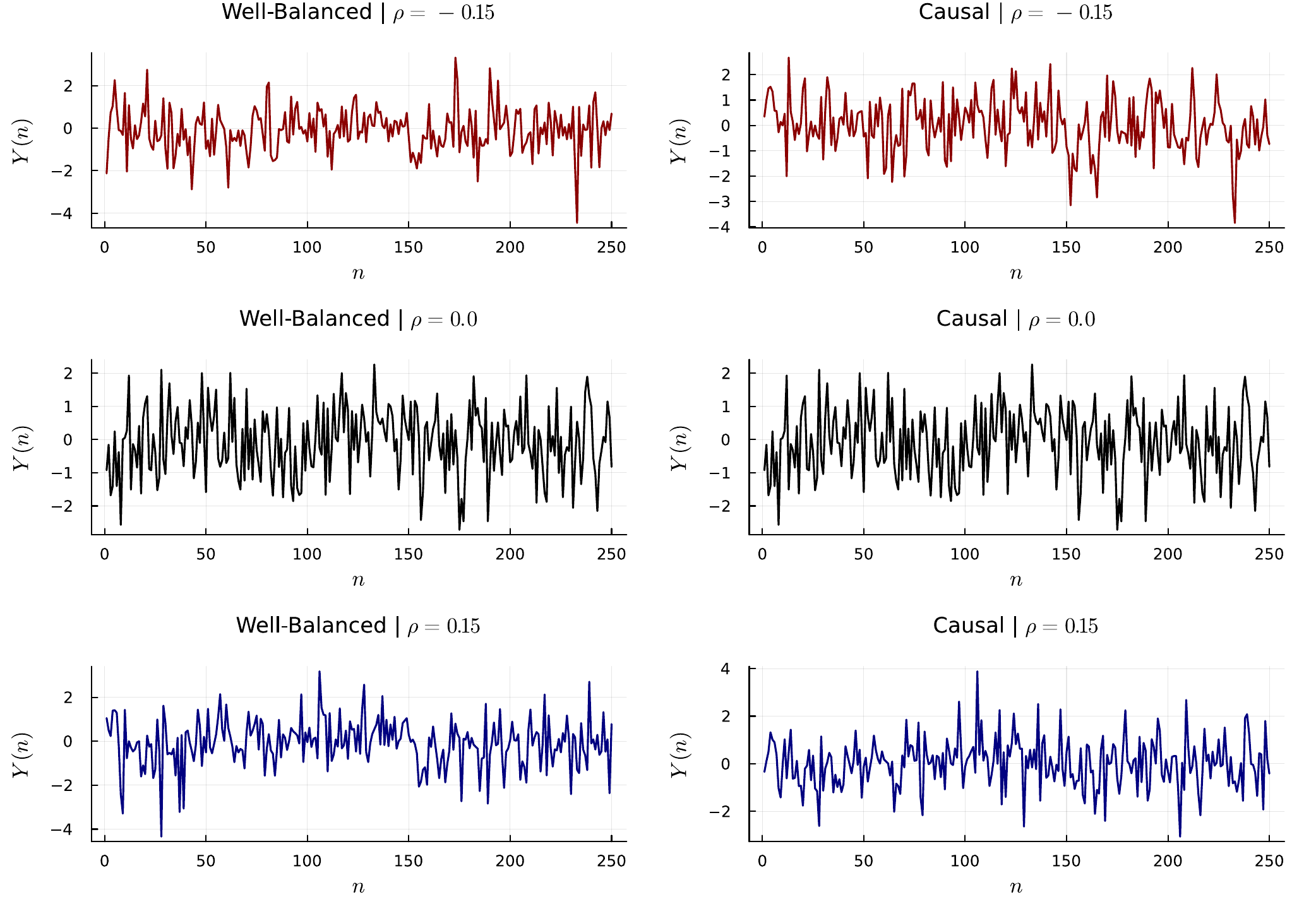}
        \caption{Example trajectories of the cfGn for different values of $\rho$, with causal and well-balanced 2d fGn as the underlying process, for $H_1=0.4$, $H_2=0.7$, $\lambda_0=0.1\pi$.}\label{fig::trajectories}
\end{figure}
\begin{figure}[ht!]
       \centering
         \includegraphics[width=1\textwidth, height=0.5\textheight]{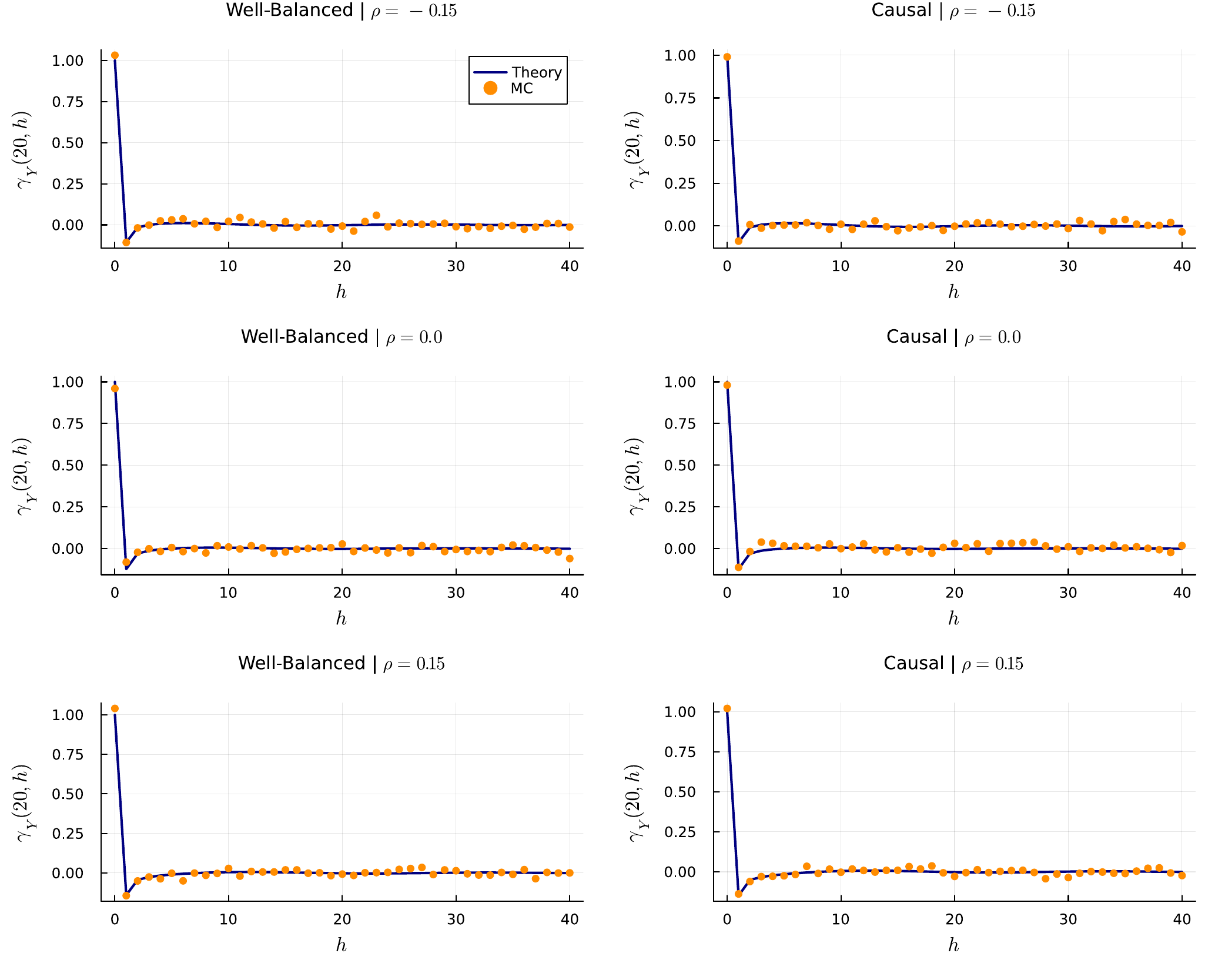}
        \caption{Comparison of theoretical ACVF given in Eq. \eqref{yautocov} and its Monte Carlo simulated counterpart for different values of $\rho$, with causal and well-balanced 2d fGn as the underlying process, for $H_1=0.4$, $H_2=0.7$, $\lambda_0=0.1\pi$, $n=20$ and $10000$ Monte Carlo simulations.}\label{fig::autocavariance}
\end{figure}
\begin{figure}[ht!]
       \centering
         \includegraphics[width=1\textwidth, height=0.5\textheight]{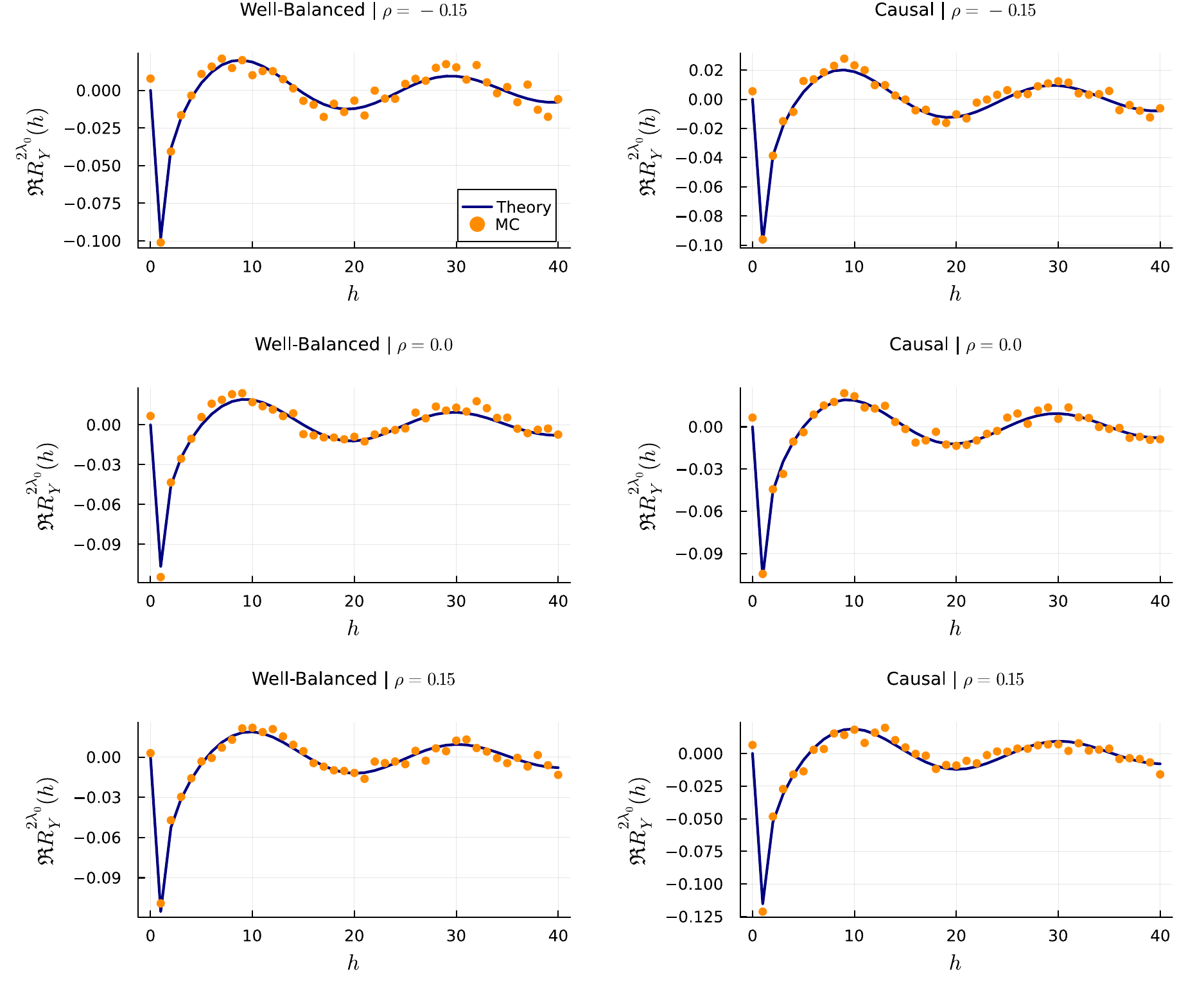}
        \caption{Comparison of theoretical real part of CAF at frequency $2\lambda_0$ given in Proposition \ref{cafprop} and its Monte Carlo simulated counterpart for different values of $\rho$, with causal and well-balanced 2d fGn as the underlying process, for $H_1=0.4$, $H_2=0.7$, $\lambda_0=0.1\pi$, $n=20$ and $10000$ Monte Carlo simulations.}\label{fig::cafre}
\end{figure}
\begin{figure}[ht!]
       \centering
         \includegraphics[width=1\textwidth, height=0.5\textheight]{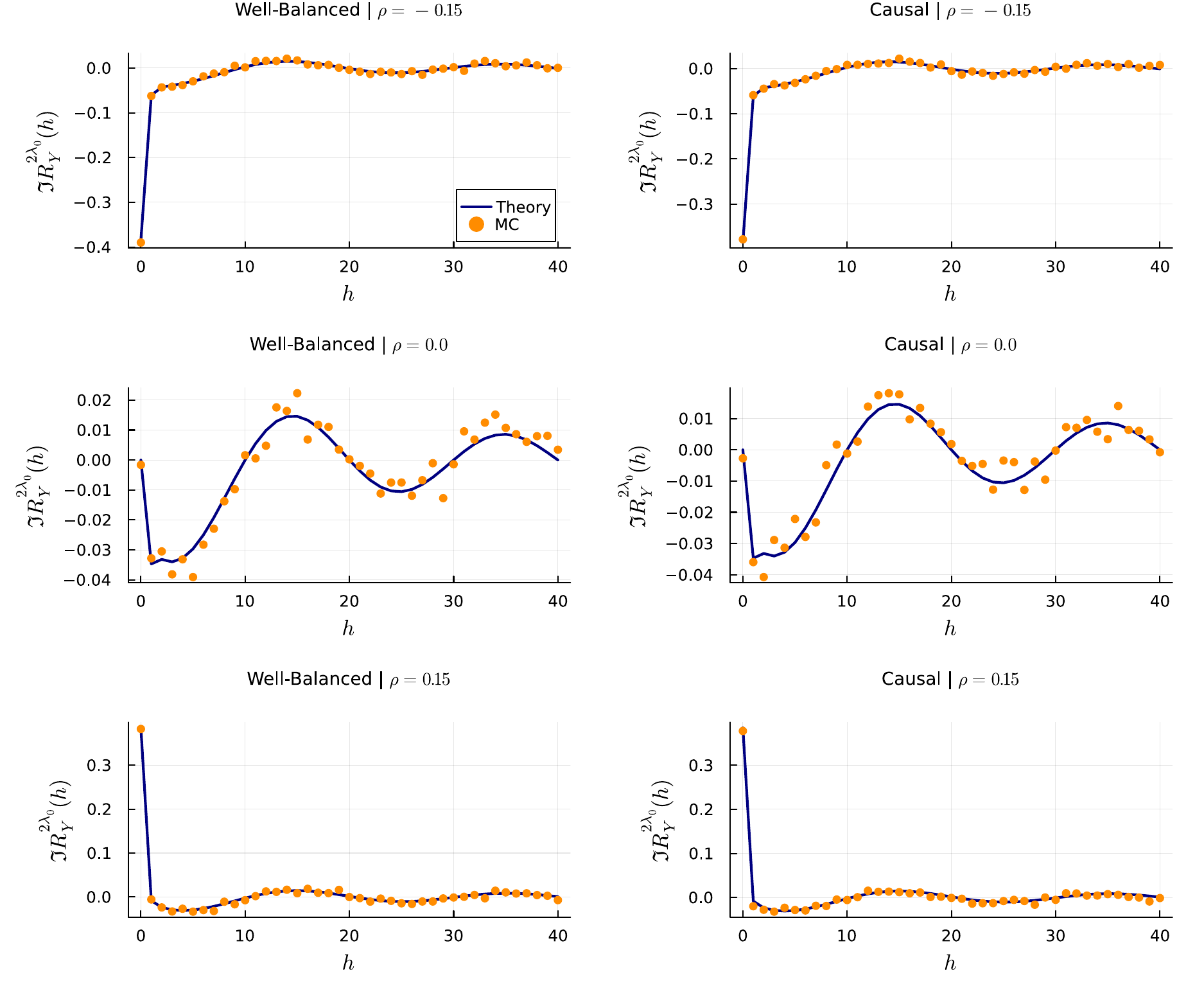}
        \caption{Comparison of theoretical imaginary part of CAF at frequency $2\lambda_0$ given in Proposition \ref{cafprop} and its Monte Carlo simulated counterpart for different values of $\rho$, with causal and well-balanced 2d fGn as the underlying process, for $H_1=0.4$, $H_2=0.7$, $\lambda_0=0.1\pi$, $n=20$ and $10000$ Monte Carlo simulations.}\label{fig::cafim}
\end{figure}
\begin{figure}[ht!]
       \centering
         \includegraphics[width=1\textwidth, height=0.5\textheight]{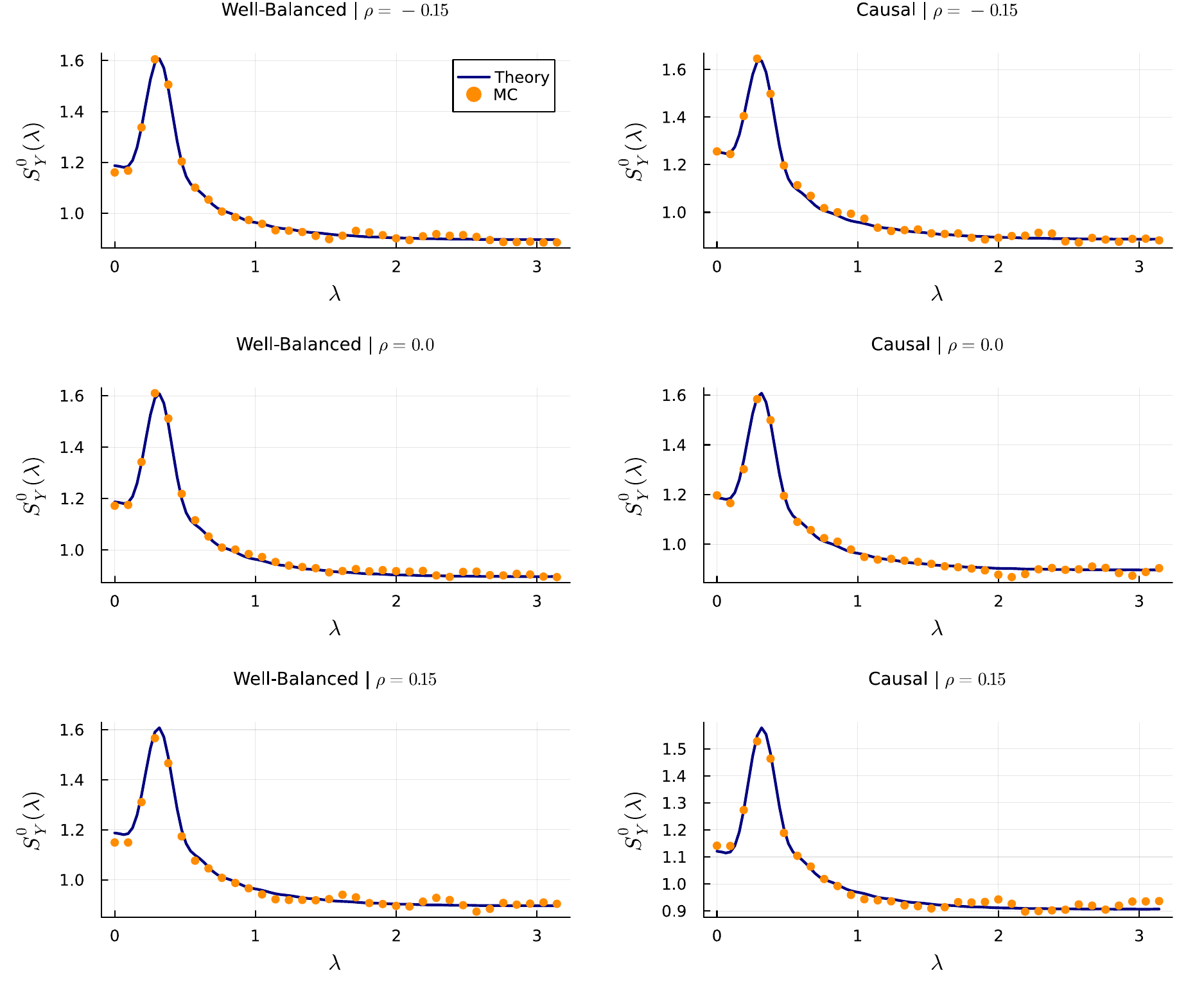}
        \caption{Comparison of theoretical cyclic spectrum at frequency $0$ given in Eq. \eqref{cs0} and ts Monte Carlo simulated counterpart for different values of $\rho$, with causal and well-balanced 2d fGn as the underlying process, for $H_1=0.4$, $H_2=0.7$, $\lambda_0=0.1\pi$, $n=20$ and $10000$ Monte Carlo simulations.}\label{fig::cs0}
\end{figure}
\begin{figure}[ht!]
       \centering
         \includegraphics[width=1\textwidth, height=0.5\textheight]{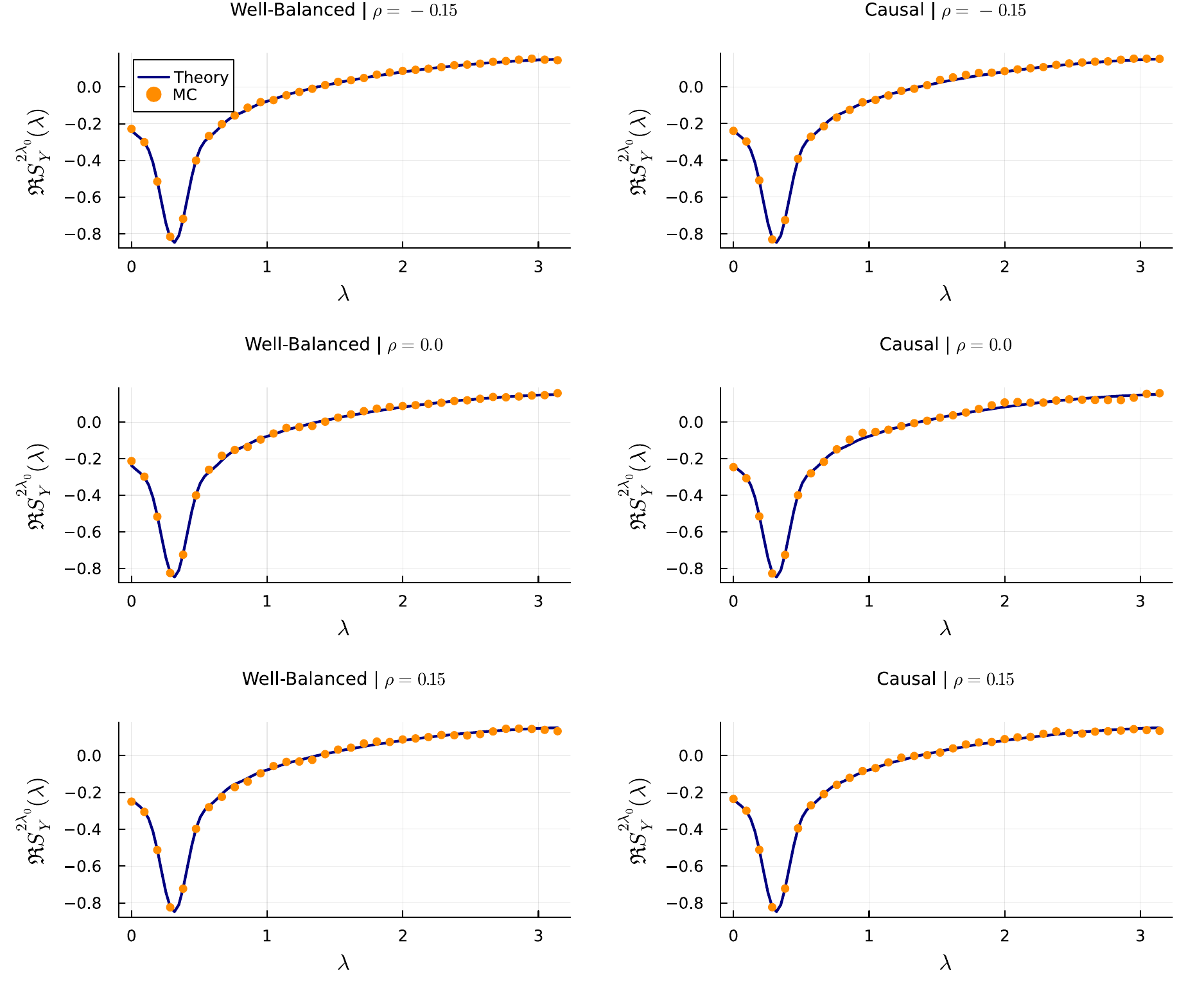}
        \caption{Comparison of theoretical real part of cyclic spectrum at frequency $2\lambda_0$ given in Eq. \eqref{cs2l0} and its Monte Carlo simulated counterpart for different values of $\rho$, with causal and well-balanced 2d fGn as the underlying process, for $H_1=0.4$, $H_2=0.7$, $\lambda_0=0.1\pi$, $n=20$ and $10000$ Monte Carlo simulations.}\label{fig::csre}
\end{figure}
\begin{figure}[ht!]
       \centering
         \includegraphics[width=1\textwidth, height=0.5\textheight]{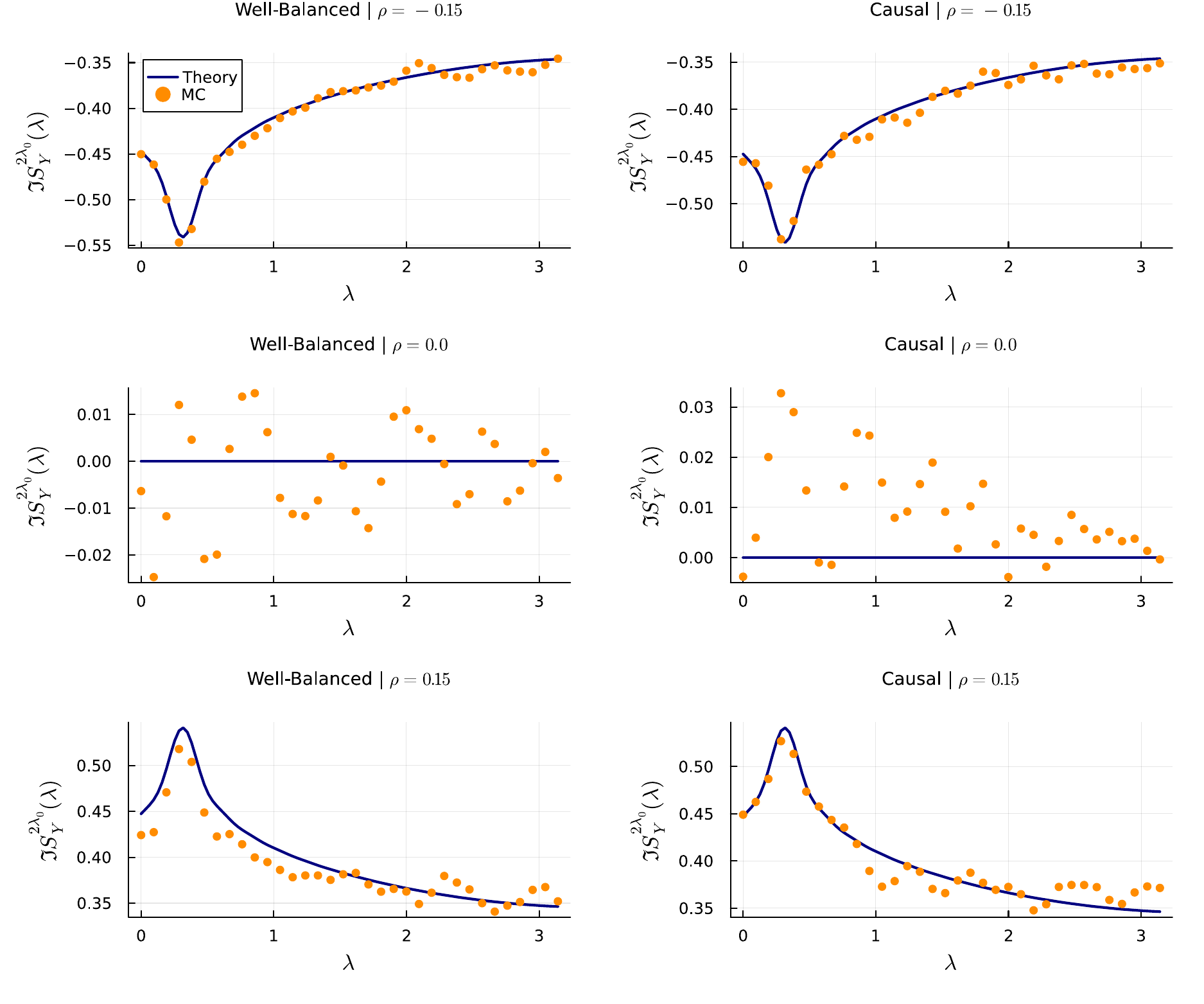}
        \caption{Comparison of theoretical imaginary part of cyclic spectrum at frequency $2\lambda_0$ given in Eq. \eqref{cs2l0} and its Monte Carlo simulated counterpart for different values of $\rho$, with causal and well-balanced 2d fBm as the underlying process, for $H_1=0.4$, $H_2=0.7$, $\lambda_0=0.1\pi$, $n=20$ and $10000$ Monte Carlo simulations.}\label{fig::csim}
\end{figure}
\begin{figure}[ht!]
       \centering
         \includegraphics[width=1\textwidth, height=1\textheight]{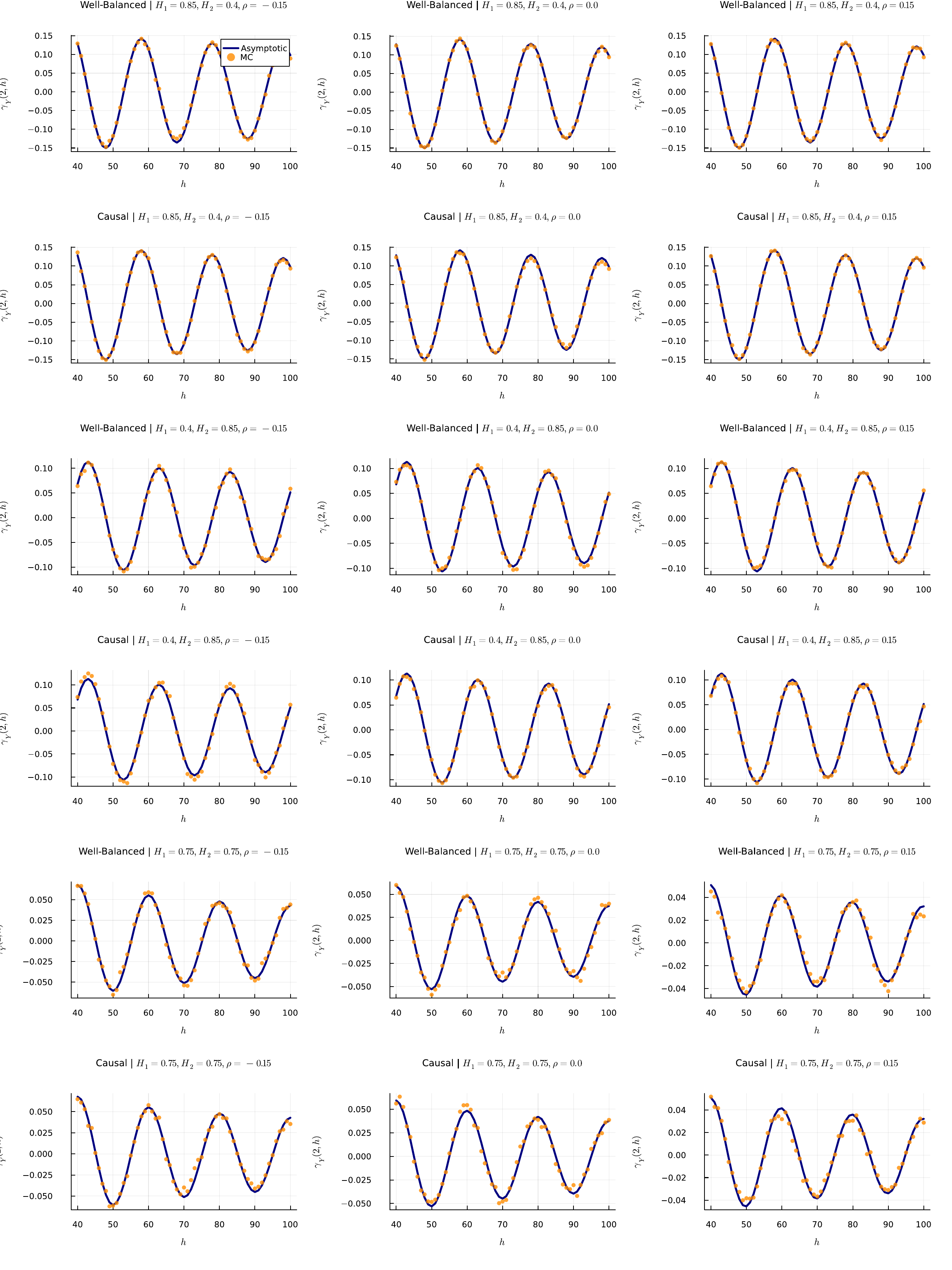}
        \caption{Comparison of theoretical asymptotic of the ACVF given in Proposition \ref{autocovasympt} and its Monte Carlo simulated counterpart for different values of $\rho, H_1, H_2$, with causal and well-balanced 2d fGn as the underlying process, for $\lambda_0=0.1\pi$, $n=2$ and $50000$ Monte Carlo simulations.}\label{fig::covasympt}
\end{figure}
\begin{figure}[ht!]
       \centering
         \includegraphics[width=1\textwidth, height=1\textheight]{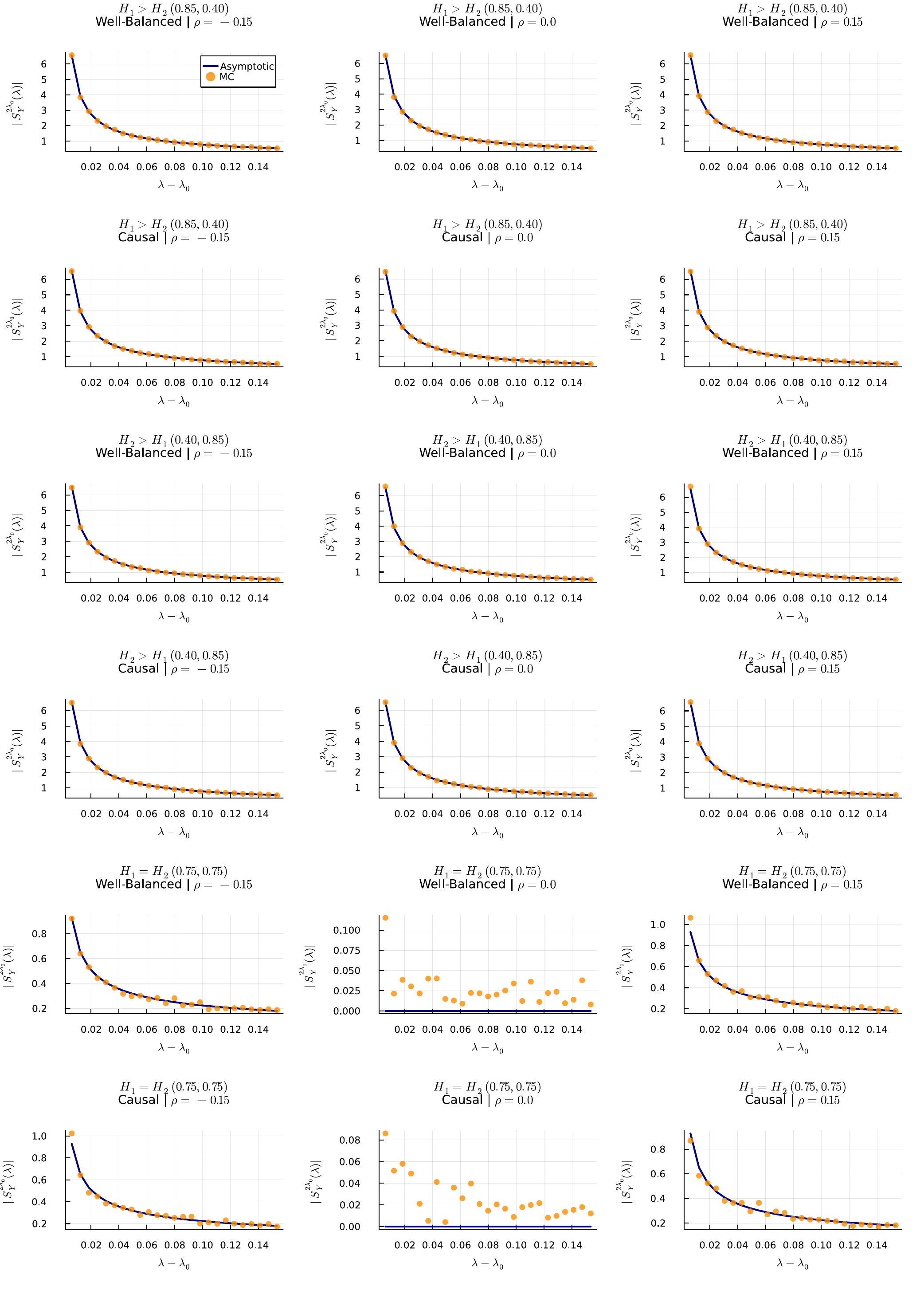}
        \caption{Comparison of theoretical asymptotic of absolute value cyclic spectrum at frequency $2\lambda_0$ given in Proposition \ref{csasympt} and its Monte Carlo simulated counterpart for different values of $\rho, H_1, H_2$, with causal and well-balanced 2d fGn as the underlying process, for
        $\lambda_0=0.1\pi$, $n=20$ and $50000$ Monte Carlo simulations.}\label{fig::csasympt}
\end{figure}
\section*{Acknowledgments}
The work was funded by the National Science Center in Poland as part of the OPUS grant: 2025/57/B/ST7/03507 „Fusion of an inspection robot-based, non-contact measurements for sound driven condition monitoring of rolling element bearings under varying speed, harsh noise and uncertain experimental conditions”.
\bibliography{bibliography}

\end{document}